\documentclass{article}
\usepackage{ijcai15}
\usepackage{graphicx}
\usepackage{amsthm,amsmath,amssymb,latexsym}

\usepackage{times}
\usepackage{helvet}
\usepackage{courier}

\def\L{\mathcal{L}}

\def\row#1#2{{#1}_1,\ldots ,{#1}_{#2}}

\newtheorem{prop}{Proposition}[section]
\newtheorem{theorem}[prop]{Theorem}
\newtheorem{example}[prop]{Example}
\newtheorem{cor}[prop]{Corollary}
\newtheorem{lemma}[prop]{Lemma}

\newtheorem{definition}[prop]{Definition}

\title{Generalizing the Single-Crossing Property on Lines and Trees to Intermediate Preferences on Median Graphs
}
\author{Adam Clearwater \\
The University of Auckland\\
Auckland, New Zealand \\
acle553@aucklanduni.ac.nz
\And
Clemens Puppe\\
Karlsruhe Institute of Technology\\
Karlsruhe, Germany \\
clemens.puppe@kit.edu
\And
Arkadii Slinko\\
The University of Auckland\\
Auckland, New Zealand \\
a.slinko@auckland.ac.nz
}

\begin{document}

\maketitle

\begin{abstract}
Demange (2012) generalized the classical single-crossing property to the intermediate property on median graphs and proved that the representative voter theorem still holds for this more general framework. We complement her result with proving that the linear orders of any profile which is intermediate on a median graph  form a Condorcet domain. We prove that for any median graph there exists a profile that is intermediate with respect to that graph and that one may need at least as many alternatives as vertices to construct such a profile.  We provide a polynomial-time algorithm to recognize whether or not a given profile is intermediate with respect to some median graph. Finally, we show that finding winners for the Chamberlin-Courant rule is polynomial-time solvable  for profiles that are single-crossing on a tree. 
\end{abstract}

\section{Introduction}

Condorcet's famous paradox %\cite{Condorcet1785} 
demonstrates that pairwise majority voting may produce
intransitive collective preferences. The question of whether, and if so, how, this problem can be overcome by means
of restrictions on the domain of admissible individual preferences has attracted constant interest over the recent
decades, see \cite{Gaertner2001} for a detailed overview.  Probably the
best-known domain restriction is single-peakedness \cite{Black1948} which is frequently employed in models of political economy. It stipulates that all alternatives can be arranged along one dimension, for instance according to the political left-right spectrum, so that each voter has an ideal choice in the set of alternatives and the alternatives that are further from their ideal choice are preferred less. 
The concept of single-peakedness itself can be generalized considerably, 
however, the sufficiency of single-peakedness for
transitivity of the (strict) majority relation is confined to the classical one-dimensional case only.\footnote{Single-peakedness
on trees still guarantees the existence of a Condorcet winner, and single-peakedness on a median graph the existence
of a ``local'' Condorcet winner \cite{BandeltBarthelemy1984}.}

A different and frequently useful sufficient condition for transitivity of the majority relation is the  single-crossing
property \cite{Mirr71}. %which stipulates that the {\em voters} can be arranged on a one-dimensional linear spectrum. 
A profile of individual preferences is said to have the single-crossing property if the {\em voters} can be arranged on a one-dimensional linear spectrum so that, for all pairs of alternatives $(a,b)$, the set of voters who prefer $a$ to $b$ and also the set of voters who prefer $b$ to $a$
are both convex.  As shown by Rothstein \shortcite{Roth91}, every single-crossing profile has a so-called {\it representative voter}, i.e.,~a voter whose (strict) preference coincides with the (strict) majority relation.\footnote{One can show that such a result does not hold for single-peaked
preferences even on a line.}

Roberts \shortcite{Roberts1977}, Blair and Crawford \shortcite{blair-craw}, Gans Smart \shortcite{GansSmart1996} provide a number of economic applications of single-crossingness among which are voting models of redistributive income taxation and trade union bargaining. All of them represent situations when preferences of individuals depend on a single parameter while in practice we may have a number of them. So there are compelling economic reasons which prompt us to consider single-crossingness on graphs more general than a line. Kung \shortcite{Kung2014}, in particular, argues that single-crossingness on trees can be applied to the study of networks.

 In contrast to the case of single-peaked preferences, the sufficiency of the single-crossing property for transitivity of the strict majority relation generalizes to median graphs \cite{Demange2012}, which, in particular, include trees and lattice graphs. 
%In the present paper, we prove the transitivity of the strict majority relation and a representative voter theorem for single-crossing profiles on trees.
%This result also follows from the analysis of  %Demange  
This generalisation is based on the notion of intermediate preferences \cite{Grandmont1978}. Assume that voters can be indexed by the vertices of a graph, and say that a profile
satisfies the {\em intermediateness property} (or simply that it is {\em intermediate}) with respect to this graph if, for all pairs of alternatives $(a,b)$ and any two voters $i$ and $j$ who prefer $a$ to $b$, all voters that lie on any shortest path between $i$ and $j$ on the graph also prefer $a$ to $b$. Evidently, if the graph is a line, %or more generally a tree, 
a profile satisfies the intermediateness property if and only if it satisfies the single-crossing property. 
%\cite{Demange2012} then proves a representative voter theorem for all median graphs. 
%The result presented here follows from this general result since, as is well known, every tree is a median graph. 
%However, our direct proof for single-crossing profiles on a tree is especially simple and illuminating. 

The purpose of this paper is two-fold. Firstly, we prove that any profile which is intermediate on a median graph gives rise to a Condorcet domain.  We give a constructive proof of the existence of an intermediate profile for any median graph with $n$ vertices (generalising and strengthening the corresponding result of \citeauthor{Kung2014} for trees). We prove that $n$ alternatives are always sufficient and that there exists a median graph (actually a tree) for which intermediate profiles with fewer than $n$ alternatives do not exist. We also give a polynomial-time algorithm that recognizes whether or not a given profile is intermediate with respect to some median graph. Finally, we prove that  the Chamberlin-Courant multi-winner voting rule on single-crossing profiles on trees has a polynomial time winner-determination problem which generalises a similar result of Skowron et al. \!\shortcite{SYFE2013} for the classical single-crossing property. The corresponding problem on median graphs remains open. It is interesting to note that for 
the single-peaked property on a tree only the egalitarian version of the Chamberlin-Courant rule remains polynomial. The classical 
utilitarian version of this rule becomes NP-hard \cite{YuCE13}.

%In papers \cite{Kung2014,CPS2014} some of these results were obtained for single-crossingness on trees. 

%\iffalse
The problem addressed in the present paper is closely related to the search of maximal Condorcet
domains \cite{AbelloJohnson84,Abello91,GR:2008,DKK:2012}; see also the survey on the topic in~\cite{mon:survey}.
Indeed, any intermediate profile on a median graph provides us with a new type of Condorcet domain (although possibly
not maximal). 
%\fi

\section{Preliminaries}

\subsection{Linear orders and profiles}

Let $A$ and $V$ be two finite sets of cardinality 
$m$ and $n$, respectively. The elements of $A$ will be called alternatives, the elements of
$V=\{1,2,\ldots, n\}$ voters. We assume that the voters have preferences over the set of alternatives.  By
$\L(A)$ we denote the set of all (strict) linear orders on $A$; they represent the preferences of
agents over $A$. The elements of the Cartesian product
$
\L(A)^n=\L(A)\times\ldots\times \L(A)\ \ \ \mbox{($n$ times)}
$
are called $n$-profiles or simply profiles. They represent the
collection of preferences of the voters from $V$ over the alternatives from $A$. If a linear order 
$R_i$ represents the preferences of the $i$-th voter, then by
$aR_ib$, where  $a,b\in A$, we denote that this agent prefers $a$ to $b$. We also denote this as $a\succ_i b$.\par\smallskip

Given a profile $R=(\row Rn)$, we say that a linear order $R_j$ is {\em between} orders $R_i$ and $R_k$ if $a\succ_i b$ and $a\succ_k b$ imply $a\succ_j b$ for every pair of alternatives $a,b\in A$. The set of all linear orders from $R$ that are between $R_i$ and $R_k$ is denoted $[R_i,R_k]$.

\begin{definition}
Let $R=(\row Rn)$ be a profile. The {\em majority relation} $M(R)$ of $R$ over $A$  is the binary relation on $A$ such that for any 
$a,b\in A$ we have $a\succeq b$ if and only if
$|\{i\mid a\succ_ib\}|\ge |\{i\mid b\succ_ia \}|$.\par
\end{definition}

We will also write $a\succ b $ if $a\succeq b$ but not $b\succeq a$ and call it the {\em strict majority relation}.  When $n$ is odd, the majority relation coincides with the strict majority relation and is a tournament on $A$, i.e., a complete and asymmetric binary relation. 
%When $n$ is even, we may have an indifference when both $a\succeq b$ and $b\succeq a$  which  we denote as $a\sim b$. 

\begin{definition}
A {\em Condorcet domain} is a set of linear orders $C\subseteq \L(A)$ such that, no matter how many voters in the profile $P$ have each of the linear orders from $C$ as their preference relation,  the strict majority relation of $P$ is  transitive. 
\end{definition}

Condorcet domains have long been of interest to Social Choice scientists and mathematicians alike; see the aforementioned survey by \citeauthor{mon:survey}. For a detailed discussion of single-crossing condition see \cite{BredereckCW13}.

\subsection{Median Graphs, Geodesic Convexity and Geodesic Betweenness}
%\setcounter{equation}{0}

%This section contains the graph-theoretic preliminaries.

%\subsection{Median Graphs}

Let $G=(V,E)$ be a connected graph. The {\em distance} $d(u,v)$ between two vertices $u,v\in V$ will be the smallest number of edges that a path from $u$ to $v$ may contain. While the distance is uniquely defined, there may be several shortest paths from $u$ to $v$. We say that vertex $w$ is {\em geodesically between} vertices $u$ and $v$ if $w$ lies on a shortest path that connects
$u$ and $v$ or, alternatively, $d(u,v)=d(u,w)+d(w,v)$. 
\begin{definition}
A {\em (geodesically) convex} set in a graph $G=(V,E)$ is a subset $C\subseteq V$ such that for any two vertices $u,v\in C$ all vertices of any shortest path between $u$ and $v$ in $G$ lie entirely in~$C$.
\end{definition}

\begin{definition}
A connected graph $G=(V,E)$ is called a {\em median} graph if for any three distinct vertices $u,v,w\in V$ there is a unique vertex $m(u,v,w)$, called {\em median}, which lies on shortest paths from $u$ to $v$, from $u$ to $w$ and from $v$ to $w$. 
\end{definition}

Trees and lattice graphs are examples of median graphs. It is known that median graphs are bipartite and hence do not contain triangles \cite{BandeltBarthelemy1984}.\par\smallskip

To describe the structure of an arbitrary median graph we remind the concept of {\em convex expansion} for graphs. %For subsets $S,T$ of the vertex-set of a graph, $E(S,T)$ denotes the set of edges connecting vertices in $S$ and vertices in $T$.

\begin{definition}
Let $G=(V,E)$ be a graph. 
Let $W_1,W_2\subset V$ be subsets such that $W_1\cup W_2=V$, $W_1\cap W_2\ne \emptyset$ and there are no edges connecting vertices  of $W_1\setminus W_2$ and vertices of $W_2\setminus W_1$. The {\em expansion} of $G$ with respect to $W_1$ and $W_2$ is the graph $G'$ constructed as follows:
\begin{itemize}
\item  each vertex $v\in W_1\cap W_2$ is replaced by two vertices $v^1$, $v^2$ joined by an edge;
\item  $v^1$ is joined to the neighbours of $v$ in $W_1\setminus W_2$ and $v^2$ is joined to the neighbours of $v$ in $W_2\setminus W_1$;
\item if $v,w\in W_1\cap W_2$ and $vw\in E$, then $v^1$ is joined to $w^1$ and $v^2$ is joined to $w^2$.
\end{itemize}
If $W_1$ and $W_2$ are convex, then $G'$ will be called a {\em convex expansion} of $G$. 
\end{definition}

\begin{example}[Convex expansion]
In the graph $G$ shown on the left of the following figure we set $W_1=\{a,b,c,d\}$ and $W_2=\{c,d,e,f\}$. These are convex and their intersection $W_1\cap W_2=\{c,d\}$ is not empty. Also the vertices of $W_1\setminus W_2=\{a,b\}$ and $W_2\setminus W_1=\{e,f\}$ have no edges between them. On the right we see the graph $G'$ obtained by the convex expansion of $G$ with respect to $W_1$ and $W_2$.
\begin{center}
\resizebox{8.4cm}{!}{\includegraphics{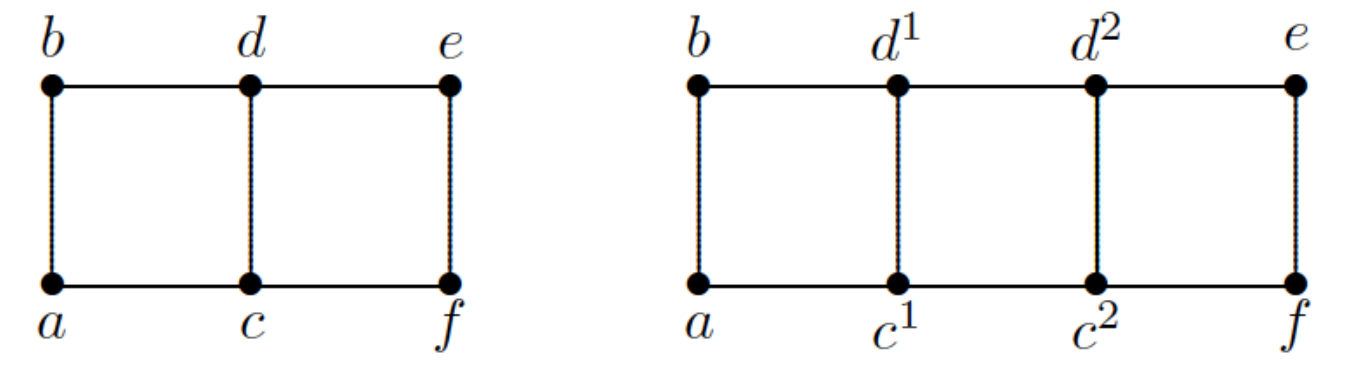}}
\end{center}
\end{example}

The following important theorem about median graphs is due to Mulder \shortcite{mulder}.

\begin{theorem}[Mulder's convex expansion theorem]
A graph is a median graph if it can be obtained from a trivial one-vertex graph by repeated convex expansions. 
\end{theorem}

%\begin{definition}
%A graph satisfies the {\em Helly property} if for any family $\mathcal H$ of convex sets in which every two sets have a non-empty intersection we have $\bigcap_{H\in \mathcal H} H \ne \emptyset$.
%\end{definition}

%\begin{proposition}
%Any median graph satisfies the Helly property, that is, for any family $\mathcal H$ of convex sets in which every two have a non-empty intersection we have $\bigcap_{H\in \mathcal H} H \ne \emptyset$. 
%\end{proposition}

%\begin{proof}
% Let us consider any three convex sets $A,B,C\in {\mathcal H}$ with non-empty pairwise intersections. Consider any vertices $u\in A\cap B$, $v\in B\cap C$ and $w\in C\cap A$. Then $m(u,v,w)\in A\cap B\cap C$ and $A\cap B\cap C$ is not empty.
%Let $\mathcal H$ be a family of convex sets in which every two sets have a non-empty intersection. Suppose $D_i\in {\mathcal H}$ and $D_1\cap \ldots \cap D_n=\emptyset$ and $n$ is minimal with this property. Then $n\ge 4$ and we may consider $A=D_1\cap \ldots \cap D_{n-2}$, $B=D_{n-1}$, $C=D_n$, all of which are non-empty and convex and have non-empty pairwise intersections. This contradicts to the statement above that any three sets have non-empty intersection.
%\end{proof}

\section{Intermediateness property}

Let $G=(V,E)$ be a graph with $V=\{1,\ldots,n\}$ and $R=(\row Rn)$ be a profile. We will consider the linear orders of $R$ as indexed by vertices of $G$.  

\begin{theorem}
\label{equiv3items}
Let $G=(V,E)$ be a graph with $V=\{1,\ldots,n\}$ and $R=(\row Rn)\in \L(A)^n$ be a profile.
The following conditions are equivalent:
\begin{enumerate}
\item[(i)] $R_j$ is between $R_i$ and $R_k$ whenever $j$ is geodesically between $i$ and $k$;
\item[(ii)]  for every ordered pair of alternatives $(a,b)\in A^2$ the set $V_{ab} =\{i\in V\mid a\succ_i b\}$ is convex in $G$;
\item[(iii)] for every shortest path $\row ik$ in $G$ between $i_1$ and $i_k$ the profile $(R_{i_1},\ldots, R_{i_k})$ is classical single-crossing profile relative to the order of voters determined by that path. 
\end{enumerate}
\end{theorem}

\begin{proof}
(i) $\Rightarrow$ (ii). Let $i,k\in V_{ab}$. Consider $j$ on any shortest path from $i$ to $k$.  Then by (i) $R_j$ is between $R_i$ and $R_k$ and since $i$ and $k$ agree on the pair $a,b$ we have $a\succ_j b$, whence $j\in V_{ab}$.

(ii) $\Rightarrow$ (iii). Let $I=\{\row ik\}$ As $I$ is a shortest path and $V_{ab}$ is convex, then $V_{ab}\cap I$ is a subpath. But $V_{ba}\cap I$ is also a subpath and this can happen only when there exists $\ell$ such that $V_{ab}=\{\row i{\ell}\}$ and $V_{ba}=\{i_{\ell+1},\ldots, k\}$.

(iii) $\Rightarrow$ (i). Follows from properties of classical single-crossing profiles.
\end{proof}

\begin{definition}
The profile $R=(\row Rn)$ is said to be {\em intermediate} on the graph $G=(V,E)$ if one of the equivalent conditions (i) - (iii) of  Theorem~\ref{equiv3items} holds.
\end{definition}

Let us consider several examples.

\begin{example}
\label{threecases}
{\bf (a) Classical single-crossingness.} A classical single-crossing profile $R=(\row Rn)$ is intermediate on a path:
 \begin{center}
\resizebox{5.6cm}{!}{\includegraphics{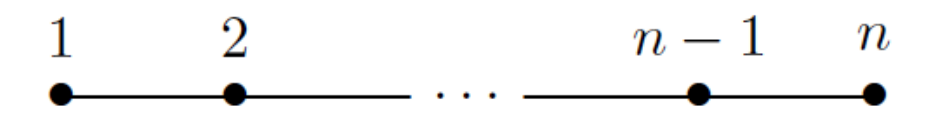}}
\end{center}

{\bf (b) Single-crossing on a tree.}
Consider the profile $R=(R_1,R_2,R_3,R_4)$ on the set $\{ a,b,c,d\}$ consisting of the following
four orders:  $R_1 = acbd$, $R_2 = abcd$, $R_3 = abdc$, and $R_4 = bacd$. 
As is easily seen, this profile is intermediate (single-crossing) on the following tree:
 \begin{center}
\resizebox{3cm}{!}{\includegraphics{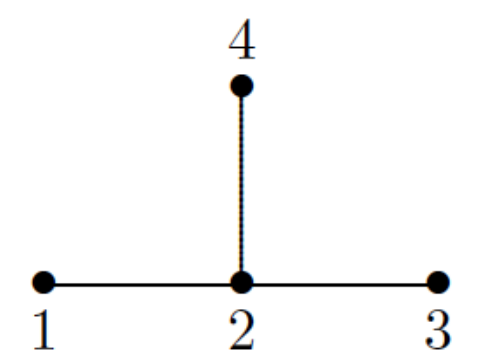}}
\end{center}

{\bf (c) Intermediate profile on a lattice graph.} The following profile $R$ is made of a set of linear orders which is a maximal Condorcet domain:
$R_1:=abcd$, $R_2:=abdc$, $R_3:=bacd$, $R_4:=badc$,
$R_5:=cdab$, $R_6:=dcab$, $R_7:=cdba$, $R_8:=dcba$.
%
%\begin{align*}
%&R_1:=abcd,\ R_2:=abdc,\ R_3:=bacd,\ R_4:=badc,\\
%&R_5:=cdab,\ R_6:=dcab,\ R_7:=cdba,\ R_8:=dcba.
%\end{align*}
It can be checked that $R$ is intermediate on the cube: 
 \begin{center}
\resizebox{3.2cm}{!}{\includegraphics{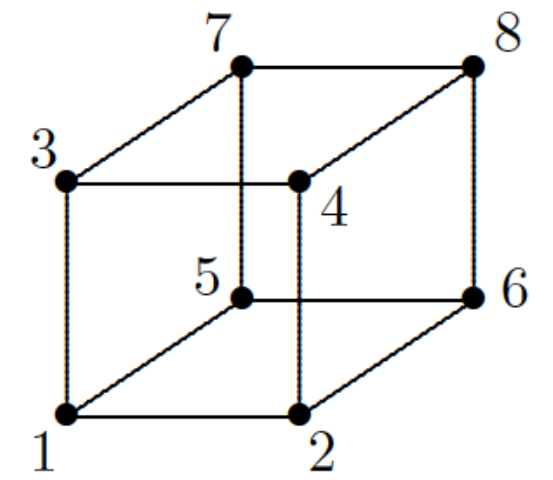}}
\end{center}
\end{example}

\begin{definition}
Let $P\in \L(A)^n$ be a profile which is intermediate on a median graph $G=(V,E)$ and $a,b\in A$. An edge $uv\in E$ is said to be an $ab$-{\em cut} if $a\succ_u b$ and $b\succ_v a$.
\end{definition}

Unlike for trees  an $ab$-cut may not be unique (although \cite{Demange2012} claims that it is). Given an edge $e$ let us denote $S(e)=\{ (a,b)\in A^2\mid \text{$e$ is an $ab$-cut}\}$.

\begin{lemma}
\label{cuts_distances}
Let $P$ be a profile which is intermediate on a median graph $G=(V,E)$. An edge $e=uv \in E$ is an $ab$-cut if and only if
$V_{ab}=\{w\in V\mid d(w,u)<d(w,v)\}$, and $V_{ba}=\{w\in V\mid d(w,u)>d(w,v)\}$. 
\end{lemma}

\begin{proof}
Let $w\in V_{ab}$. Then, due to intermediateness, among $w,u,v$  only $u$ can be the median. Hence $wuv$ is the shortest path from $w$ to $v$ and $d(w,u)=d(w,v)-1$ and $d(w,u)<d(w,v)$. The converse is clear.
\end{proof}

\begin{cor}
\label{layers}
Let $P$ be a profile which is intermediate on a median graph $G=(V,E)$. Suppose $\row ek\in E$ are $ab$-cuts. Then $S(e_1)=\ldots=S(e_k)$.
\end{cor}

\begin{proof}
Suppose $e_1=u_1v_1$ is an $ab$-cut and also a $cd$-cut. Then by Lemma~\ref{cuts_distances} $V_{ab}=V_{cd}$ and $V_{ba}=V_{dc}$. Since $e_i=u_iv_i$ is an $ab$-cut, then $u_i\in V_{ab}=V_{cd}$ and $v_i\in V_{ba}=V_{dc}$.  Hence $e_i$ is also a $cd$-cut.
\end{proof}

%The latest examples very instructive. In fact, the following is true.
We need the following technical lemma now.

\begin{lemma}
\label{l1}
Let $P=(P_1^{k_1},\ldots, P_n^{k_n})$ be a profile, where linear order $P_i$ is repeated $k_i$ times and $P_i\ne P_j$ if $i\ne j$.  If  $P$ is intermediate on a median graph $G$, then the  profile $\hat{P}=(\row Pn)$ is also intermediate on a median graph.
\end{lemma}

\begin{proof}
Any two identical linear orders are either neighbors in $G$ or connected by a path of vertices with identical linear orders. Contracting them to a single vertex results in the graph sought for.
\end{proof}

Let $P$ be a profile. By $\mathcal{D}(P)$ we denote the set of all distinct linear orders present in $P$. If all linear orders of $P$ are different, we call the profile $P$ {\em reduced}.

\begin{theorem}
\label{cd_theorem}
Let $P$ be a profile which is intermediate with respect to a median graph $G=(V,E)$. Then  $\mathcal{D}(P)$ is a Condorcet domain.
\end{theorem}

\begin{proof}
By Lemma~\ref{l1} we may assume that the profile $P$ is reduced. Let $Q=(Q_1^{s_1},\ldots, Q_m^{s_m})$ be a profile with $Q_i\in \mathcal{D}(P)$ and $Q_i\ne Q_j$ for all $i\ne j$.  We can add linear orders $Q_{m+1},\ldots, Q_n$ so that the extended profile $ \widehat{Q}=(\row Qn)$ is intermediate on a median graph $G=(\widehat{V},E)$, where $\widehat{V}=\{1,\ldots,n\}$ (renumeration of vertices may be necessary). We will also denote $V=\{1,\ldots,m\}$. 

Suppose  that  $a\succ b$ and $b\succ c$, where $\succ $ is the majority relation for $Q$. Then we have two partitions of $V$, namely, $V=V_{ab}\cup V_{ba}$ and $V=V_{bc}\cup V_{cb}$. We have $|V_{ab}|>|V_{ba}|$ and $|V_{bc}|> |V_{cb}|$. Obviously, $V_{abc}=V_{ab}\cap V_{bc}\ne \emptyset$ since each of these sets contains a majority of voters in $V$. We have $V_{ab}\subseteq \widehat{V}_{ab}$ and $V_{bc}\subseteq \widehat{V}_{bc}$, where $\widehat{V}_{xy}=\{i\in \widehat{V}\mid x\succ_i y\}$. We cannot claim that $\widehat{V}_{ab}$ or $\widehat{V}_{bc}$ contains more than half of all elements of $\widehat{V}$ but we know that $\widehat{V}_{ab}\cap\widehat{V}_{bc}\supseteq V_{abc}\ne \emptyset$. 

Consider now $\widehat{V}_{ac}$ and $\widehat{V}_{ca}$. Suppose there are several edges $e_s=u_sv_s$ ($s=1,\ldots, t$), connecting them; all of them are $ac$-cuts. Suppose $\widehat{V}_{ac}\cap \widehat{V}_{ab}\ne \emptyset$ and $\widehat{V}_{ca}\cap \widehat{V}_{ab}\ne \emptyset$. Then for any $p\in \widehat{V}_{ac}\cap \widehat{V}_{ab}$ and $q\in \widehat{V}_{ca}\cap \widehat{V}_{ab}$ the shortest path between them includes one of the $ac$-cuts, say $e_i=u_iv_i$. This means that $a\succ_{u_i} b$ and  $a\succ_{v_i} b$. But then it is not possible to have $\widehat{V}_{ac}\cap \widehat{V}_{bc}\ne \emptyset$ and $\widehat{V}_{ca}\cap \widehat{V}_{bc}\ne \emptyset$ since in this case we would have $b\succ_{u_i} c$ and  $b\succ_{v_i} c$ and by transitivity $a\succ_{u_i} c$ and  $a\succ_{v_i} c$ which contradicts to $e_i$ being an $ac$-cut. Since $\emptyset \ne V_{abc}\subseteq \widehat{V}_{ac}\cap \widehat{V}_{bc}$, we have $\widehat{V}_{ca}\cap \widehat{V}_{bc}= \emptyset$ and $V_{bc}\subseteq V_{ac}$.

This means $a\succ c$ and $\succ$ is transitive.
\end{proof}

%Firstly,  we give a short proof of the so-called Representative Voter Theorem.

%\begin{theorem}[\cite{Demange2012}]
%\label{RVT}
%Let $n$ be odd. If a profile $P=(\row Pn)$  is single-crossing with respect to a tree $T=(V,E)$, then there exists $i\in \{1,\ldots,n\}$ such that  preference order $P_i$  coincides with  the majority relation. % $\succ $ associated with $P$. %In particular, the majority relation is transitive. 
%\end{theorem}

%\begin{proof}
%Firstly we note that there is a natural absract convexity on trees \cite{EdelmanJamison1985}: the set is called convex if it is connected.   It is easy to check that it satisfies the {\em Helly property} \cite{Bollobas1986}: if for any family $\mathcal H$ of convex sets in which every two sets have a non-empty intersection we have $\bigcap_{H\in \mathcal H} H \ne \emptyset$.

%Let, as above, $V_{xy}$ be the set of voters who prefer $x$ to $y$. These sets are convex.
%Consider the set of subsets ${\mathcal M}= \{V_{xy}\mid x\succ y\}$, where $\succ $ is the majority relation. Any two subsets $V_{xy}, V_{zt}\in {\mathcal M}$ have a nonempty intersection since each of them contains a majority of all voters. By the Helly property we have
%$
%\bigcap_{V_{xy}\in {\mathcal M}} V_{xy}\ne \emptyset.
%$
%Voters'  preferences in this intersection coincide with the majority relation.
%\end{proof}

Now we are going to prove that for any median graph we can construct a reduced profile which is intermediate on that particular graph. We then discuss how many alternatives may be needed.

\begin{theorem}  
\label{c-m-theorem}
For every median graph $G=(V,E)$ with $|V|=n$ there exists  a reduced preference profile $R=(\row Rn)\in \L(Y)^n$ on a set of alternatives $Y$ with $|Y|\le n$ such that $R$ is intermediate on $G$.
\end{theorem}

\begin{proof}
Since the statement is true for the trivial graph consisting of a single vertex, arguing by induction, we assume that  the statement is true for all median graphs with  $k$ vertices or less.  Let $G'=(V',E')$ be a median graph with $|V'|=k+1$. By Mulder's theorem $G'$ is a convex expansion of  a certain median graph $G=(V,E)$ relative to convex subsets $W_1$ and $W_2$, where $|V|=\ell\le k$. By induction there exists a reduced profile $R=(\row R{\ell})\in \L(X)^\ell$ with $|X| \le k$ which is intermediate on $G$.

To obtain a new profile $R'$ which is intermediate on $G'$ we clone an arbitrary alternative $x\in X$ and introduce a clone $y\notin X$ of $x$.\footnote{We say that $x$ and $y$ are clones if they are neighbours in any linear order of the domain (cf. Elkind et al (2011) and (2012)).}  The  linear orders of the new profile $R'$ will be constructed as follows.  If $v$ is a vertex of $W_1\setminus W_2$ to obtain $R'_v$ we replace $x$ with $xy$ in $R_v$ placing $y$ lower than $x$ and to obtain $R'_u$ for  $u\in W_2\setminus W_1$ we replace $x$ with $yx$ in $R_u$ placing $y$ higher than $x$. 
Let  $v$ now be in $W_1\cap W_2$. In the convex expansion this vertex is split into $v^1$ and $v^2$. To obtain $R'_ {v^1}$ we clone the linear order $R_v$ replacing $x$  with $xy$ and to obtain $R'_ {v^2}$ we clone the same linear order $R_v$ replacing $x$  with $yx$. We have $Y=X\cup \{y\}$ so the number of alternatives has increased by one, so it is not greater than $k+1=|V'|$. The profile $R'$, so obtained, is also reduced.

To prove that $R'$ is intermediate on $G'$. 
we need to consider several cases. For example, let $u\in W_1\setminus W_2$ and $v\in W_2\setminus W_1$ and let us prove that all linear orders on a shortest path between $u$ and $v$ are between $R'_u$ and $R'_v$. Firstly, we note that any such shortest path will contain an edge $e=w^1w^2$, where $w\in W_1\cap W_2$. There is the corresponding path in $G$, where the edge $e$ is contracted to $w$ and all linear orders on that path are between $R_u$ and $R_v$. The way we placed $y$ in these linear orders will not disturb the betweenness since all linear orders $R'_z$ between $R'_u$ and $R_{w^1}$ (inclusive) will have  $x\succ'_z y$ and all others will have  $y\succ'_z x$. All other cases are considered similarly. 
\end{proof}

The constant $n$ in this theorem  cannot be improved even for trees. However we must exclude some trivial cases. For example, if all the linear orders in a profile are equal, then it is intermediate for any graph. So we have to restrict ourselves to reduced profiles. 

\begin{theorem}
For the star $S_n$ with $n$ vertices (the complete bipartite graph $K_{1,n{-}1}$)  there does not exist a reduced profile $R=(\row Rn)\in \L(X)^n$ on the set $X$ of alternatives of cardinality smaller than $n$ which is intermediate on $S_n$.
\end{theorem}

\begin{proof}
We reason by induction on $n$.  The base case is the star $S_2$ with two vertices for which the result is clear. 

%Consider the star graph $S_n$ on $n$ vertices whose one vertex has degree $n-1$ and all others are leaves.  Reasoning by induction 

Let us assume that for $S_n$, $n\ge 2$, we cannot construct a reduced profile with less than $n$ alternatives which is intermediate on $S_n$. Consider $S_{n+1}$ and suppose towards a contradiction that we can find a reduced profile $P=(\row P{n+1})$ with $n$ alternatives which is intermediate on $S_{n+1}$. Suppose that the vertices of $S_{n+1}$ are numbered so that voter 1's vertex has degree $n$ and voter 1 has preferences expressed by the linear order 
$a_1\succ_1 a_2\succ_1\ldots \succ_1 a_{n}$.  Suppose, first, that $s$ is the smallest number for which one of the edges, say the one connecting $1$ with $n+1$, is an $a_1a_s$-cut  for $S_{n+1}$; if such a number does not exist, $a_1$ can be removed from the profile and it will stay reduced. Then $a_s\succ_{n+1} a_1$ and $a_1\succ_i a_s$ for all $i\le n$ (otherwise intermediateness fails). Let us remove now vertex $n+1$ from the graph and linear order $P_{n+1}$ from the profile. Then we get graph $S_n$ and the profile $P'=(\row Pn)$ which is intermediate on $S_n$.  No edge in $S_n$ is an $a_1a_s$-cut now. %Let $t>s$ be the smallest number for which $a_1a_t$-cut exist for $P'$ (if none, $a_1$ can be removed). 
Let us now remove the alternative $a_1$ from $P'$ to obtain a profile $P''=(P''_1,\ldots, P''_n)$. Then $P''$ is still intermediate on the star $S_n$ and we claim that $P''$ is reduced. If not, then after removal of $a_1$, at least two linear orders, say $P''_i$ and  $P''_k$ become equal. If $i=1$, then the only cut the edge $1k$ had was an $a_1a_t$-cut for some $t>s$. Then we had $a_t\succ_k a_1\succ_k a_s$ while $a_s\succ_1 a_t$; so we see that $1k$ had also a $a_sa_t$-cut. This contradicts $P''_1=P''_k$. % that proves the theorem. 

Now suppose $P''_i=P''_k$ with $i\ne 1$ and $k\ne 1$. Then $P_i$ and  $P_k$ must agree on $a_2,\ldots, a_n$. Due to intermediateness they must also agree with $P_1$ on these, i.e., $ a_2\succ_i\ldots \succ_i a_{n}$ and $ a_2\succ_k\ldots \succ_k a_{n}$.  Since $a_1$ cannot be on the top of each of them (otherwise one would be equal to $P_1$) we have $a_2\succ_i a_1$ and $a_2\succ_k a_1$, which contradicts intermediateness.

Hence $P''$ is reduced. Since $P''$ has $n-1$ alternatives, this is a contradiction.
\end{proof}

\section{Algorithmic aspects of single-crossedness}

\subsection{A recognition algorithm.}

The goal of this section is to give a polynomial-time algorithm for recognising intermediate profiles on median graphs. 
We need to know more about the structure of intermediate profiles on median graphs.

\begin{theorem}
Let $R=(\row Rn)\in \L(A)^n$ be a profile whose linear orders are indexed by the vertices of graph $G=(V,E)$. Let $a,b\in A$ be a pair of alternatives such that $\emptyset \ne V_{ab}\ne V$. Let also $e_i=u_iv_i$, $i=1,\ldots,s$ be the edges that connect $V_{ab}$ and $V_{ba}$ with $u_i\in V_{ab}$ and $v_i\in V_{ba}$. Then $G$ is a median graph and $R$ is intermediate on $G$ iff
\begin{enumerate}
\item[(i)] $\row us$ are all distinct and so are $\row vs$;
\item[(ii)] $S(e_1)=\ldots = S(e_s)$;
\item[(iii)] Induced graphs on $V_{ab}$ and $V_{ba}$ are median graphs in their own right;
\item[(iv)] $R_{ab}=\{R_u\mid u\in V_{ab}\}$ and $R_{ba}=\{R_v\mid v\in V_{ab}\}$ are the profiles which are intermediate on $V_{ab}$ and $V_{ba}$, respectively;
\item[(v)] Let $\bar{G}$ be the graph $G$ with edges $\row es$ contracted. Let $\bar{V}_{ab}$ and $\bar{V}_{ba}$ be images of $V_{ab}$ and $V_{ba}$ in $\bar{G}$. %(now these will have non-empty intersection). 
Then $G$ is the convex expansion of $\bar{G}$ with respect to $\bar{V}_{ab}$ and~$\bar{V}_{ba}$.
\end{enumerate}
\end{theorem}

\begin{proof}
$\Longrightarrow$ 
(i) Suppose $v_i=v_j$, then $u_i\ne u_j$ and $(a,b)\in S(e_i)=S(e_j)$. Since $G$ is bipartite it has no triangles, hence $u_i v_i u_j$ is the shortest path between $u_i$ and $u_j$.  This contradicts to intermediateness since $R_{u_i}$ and $R_{u_j}$ agree on $(a,b)$ but $R_{v_i}$ disagrees with them.

(ii) is Corollary~\ref{layers}.

(iii) Follows from an observation that a shortest path between vertices $w, w'\in V_{ab}$ cannot involve any of $\row es$ and hence lies entirely in $V_{ab}$. 

(iv) and (v) are now obvious.

$\Longleftarrow$ Firstly, $G$ is median due to Mulder's theorem. To check intermediateness, let us consider a shortest path between $u$ and $v$ from $V$. If they are both in $V_{ab}$ or $V_{ba}$, this case is clear. Suppose  $u\in V_{ab}$ and $v\in V_{ba}$ and they agree on $(c,d)$ for some $c,d\in A$. The shortest path between $u$ and $v$ contains one of the edges $\row es$, say $e_i$ so it goes through vertices $u,u_i,v_i,v$.  Since elements of $V_{cd}$ are found both in $V_{ab}$ and $V_{ba}$ by Lemma~\ref{cuts_distances} we cannot have $(c,d)\in S$, where $S=S(e_1)=\ldots = S(e_s)$. Thus both $c\succ_{u_i} d$ and $c\succ_{v_i} d$. The part of the path connecting $u$ and $u_i$ is the shortest path in $V_{ab}$ so all linear orders on this path agree with $u$ and $u_i$ on $(c,d)$. So do the linear orders on the path from $v_i$ to $v$. This proves the theorem.
\end{proof}

The idea of the recognition algorithm is now clear. We give the construction only for a reduced profile $R=(\row Rn)$. The order of linear orders in the profile can be ignored so we actually deal with the domain ${\mathcal D}=\{\row Rn\}$. Firstly, we identify the graph $G_{\mathcal D}$ to be tested.  For this we find all pairs of `neighboring' linear orders. Any two linear orders $P,Q\in {\mathcal D}$ define the `interval' $[P,Q]$ as the set of all linear orders in ${\mathcal D}$ which are between $P$ and $Q$ and call $P$ and $Q$ {\em neighbors} if $[P,Q]=\{P,Q\}$. We draw edges between the neighboring orders and obtain the graph  $G_{\mathcal D}$ on linear orders from $\mathcal D$. The construction of this graph requires $O(m^2n^3)$ operations, where $m$ is the number of alternatives. 

If $R$ was intermediate on a median graph $G$, then  $G_{\mathcal D}$ is exactly $G$. Indeed, if $u$ and $v$ were neighbors in $G$, then $[R_u,R_v]=\{R_u,R_v\}$. If only $R_w\in [R_u,R_v]$ for $w\ne u$ and $w\ne v$, then, since $u$ and $v$ are neighbors, the median $m(u,v,w)$ is either $u$ or $v$. Suppose $m(u,v,w)=u$. Then, due to intermediateness, $R_u$ is between $R_w$ and $R_v$ from which $R_w=R_u$. On the other hand, if $u$ and $v$ are not neighbors in $G$, then it is easy to see that $[R_u,R_v]\ne \{R_u,R_v\}$.

We then pick any pair $(a,b)\in A$ for which $V_{ab}$ and $V_{ba}$ are both nonempty and select edges $e_i=u_iv_i$, $i=1,\ldots,s$  of $G_{\mathcal D}$. If all $u_i$'s and all $v_i$'s are different, and $S(e_1)=\ldots =S(e_s)$, then the question of whether or not $R$ is intermediate on $G_{\mathcal D}$ will be reduced to the questions of whether or not $R_{ab}=\{R_u\mid u\in V_{ab}\}$ and $R_{ba}=\{R_v\mid v\in V_{ba}\}$ are intermediate on graphs induced on $V_{ab}$ and $V_{ba}$, respectively.  This can be arranged as a recursive algorithm.  We have proved

\begin{theorem}
\label{conmintree}
For an input profile with $n$ voters and $m$ alternatives we can determine in time polynomial in $m,n$ whether or not the given profile is intermediate on some median graph, and, if so, construct this graph.
\end{theorem}

\subsection{Chamberlin-Courant rule %of fully proportional representation
}

This section deals with single-crossing profiles on trees.

Given a society of $n$ voters $V$ with preferences over a set of $m$ candidates $A$ and a fixed positive integer $k\le m$, a method of fully proportional representation outputs a $k$-member committee (e.g., parliament), which is a subset of $A$, and assigns to each voter a candidate that  will represent this voter in the committee. The fully proportional representation rules, suggested by Chamberlin and Courant  \shortcite{cha-cou:j:cc} and Monroe \shortcite{mon:j:monroe}, have been widely discussed in Political Science and Social Choice literature alike   
\cite{pot-bra,mei-pro-ros-zoh:j:multiwinner,bet-sli-uhl:j:mon-cc}. 

It is well-known that on an unrestricted domain of preferences both rules are intractable in the classical  \cite{mei-pro-ros-zoh:j:multiwinner,bou-lu:c:chamberlin-courant} and parameterized complexity  \cite{bet-sli-uhl:j:mon-cc} senses.  \cite{SYFE2013}, however, showed that for the classical single-crossing  elections the winner-determination problem for the Chamberlin-Courant rule is polynomial-time solvable for every dissatisfaction function for both the utilitarian  \cite{cha-cou:j:cc} and egalitarian \cite{bet-sli-uhl:j:mon-cc} versions of the rule. They also generalized this result to elections with bounded single-crossing width proving fixed-parameter tractability %\cite{dowfel} 
of the Chamberlin-Courant  rule with single-crossing width as parameter.  The concept of single-peaked width was defined in \cite{CornazGS12}. 

Here we will prove that polynomial solvability remains for single-crossing profiles on any tree. But, firstly, we will remind the reader the definitions needed for discussing the Chamberlin-Courant rule. By $\text{pos}_v(c)$ we denote the position of the alternative $c$ in the ranking of voter $v$; the top-ranked alternative has position 1, the second best has position 2, etc.

\begin{definition}
\label{def:mf}
Given a profile $P$ over set~$A$  of alternatives, a mapping $r\colon P\times A \to \mathcal{Q}^+_0$
is called {\em a misrepresentation function} if for any voter $v\in V$
and candidates $c,c'\in A$ the condition
$\text{pos}_v(c)<\text{pos}_v(c')$ implies $r(v,c)\le r(v,c')$. 
\end{definition}
In the classical framework the misrepresentation of a candidate
for a voter is a function of the position of the candidate in the preference
order of that voter given by ${\bf s}=(\row s{m})$, where $0=s_1\le s_2\le \ldots \le s_m$, that is, the misrepresentation function in this case will be
$
r(v,c)=s_{\text{pos}_v(c)}.
$
Such a misrepresentation function is called {\em positional}.
An important case is the \emph{Borda misrepresentation  function} defined by the vector
 $(0,1,\ldots, m-1)$ which was used in  \cite{cha-cou:j:cc}. 
We assume that the misrepresentation function is defined for any number of alternatives and that it is polynomial-time computable.

In the approval voting framework, if a voter is represented by a
 candidate whom she approves, her misrepresentation is  zero, otherwise it is equal to one.  This function is called
 the \emph{approval misrepresentation function}.  It does not have to be positional since
 different voters may approve different numbers of candidates. 
In the general framework the misrepresentation function may be
arbitrary.  

By $w\colon V\to A$ we denote the function that assigns voters to representatives, i.e., under this assignment voter $v$ is represented by candidate $w(v)$. If $|w(V)|\le k$ we call it a $k$-{\em assignment}. The total misrepresentation $\Phi(P,w)$ of the given election under~$w$ is then given by
$
\Phi(P,w)=\ell( r(v,w(v))),
$
%$
%\Phi(P,w)=\sum_{v \in V} r(v,w(v))\quad \text{or}\quad \Phi(P,w)= \max_{v \in V} r(v,w(v))
%$
where $\ell$ is used to mean either the sum or the maximum of a given list of values  depending on the utilitarian or egalitarian model, respectively.

The Cham\-ber\-lin-Courant rule takes the profile and the number of representatives to be elected $k$ as input and outputs an optimal $k$-assignment  $w_\text{opt}$ of voters to representatives that minimizes  the total misrepresentation $\Phi(P,w)$.  %We will prove the following theorem.

\begin{theorem}
\label{CCtheorem}
For any polynomial-time computable dissatisfaction function, every positive integer $k$, and for both utilitarian and egalitarian versions of the Chamberlin-Courant rule, there is a polynomial-time algorithm that given a profile $P=(\row Pn)$ over a set of alternatives $A=\{\row am\}$, which is single-crossing with respect to some tree,  finds an optimal $k$-assignment function $w_\text{opt}$ for~$P$.
\end{theorem}

%We need to prove the following lemma first.

\begin{lemma}
\label{termsubtr}
Let $P=(\row Pn)$ be a reduced profile over a set~$A=\{\row am\}$ of alternatives, where $P$ is single-crossing with respect to a  tree $T$. Let $w_\text{opt}$ be an optimal $k$-assignment for $P$.  Let voter 1 be an arbitrary vertex of $T$, and let $b\in A$ be the least preferred alternative of voter 1 in $w_\text{opt}(V)$. Then the vertices of $w_\text{opt}^{-1}(b)$ are vertices of a terminal subtree (a subtree whose vertex-complement is also a subtree) of~$T$.
\end{lemma}

\begin{proof}
On a tree $T$ we may define the distance between any two vertices $u$ and $v$ which is the number of edges on the unique path connecting these two vertices. 
Let $v$ be the closest vertex to $1$ such that $w_\text{opt}(v)=b$. Let $u$ be the vertex on that path which is one edge closer to $1$ (it can be actually $1$ itself). Due to minimality of $T$ the edge $(u,v)$ is an $(a,b)$-cut for some $a= w_\text{opt}(u)$ where $a\succ_1 b$. Let us show that $w_\text{opt}^{-1}(b)=V_{ba}$. Suppose first that $w_\text{opt}(v')=b$. Then $b\succ_{v'} a$ (otherwise $v'$ would be assigned $a$) and hence $v'\in V_{ba}$. Suppose now $v'\in V_{ba}$. Then $v'$ is connected by a path within $V_{ba}$, hence the unique path between 1 and $v'$ passes through $v$. Since linear orders on this path form a classical single crossing subprofile we have $w_\text{opt}(v')=b$.
\end{proof}

The fact that the least preferred alternative of voter 1 in the elected committee $w(V)$ represents voters in a terminal subtree of $T$ is important in our design of a dynamic programming algorithm. We will fix an arbitrary leaf, without loss of generality it will be voter 1, and reduce the problem of calculating an optimal assignment for a profile $P=(\row Pn)$ over a set of alternatives $A=\{\row am\}$ to a partially ordered set of subproblems which will be defined shortly. Let us denote the original problem as $(P,A,k)$. We define the set of subproblems as follows. A triple $(P',A',k')$, where $k'\le k$ is a positive integer, $P'$ is a subprofile of $P$ and $A'$ is a subset of $A$ is a subproblem of $(P,A,k)$ if 
\begin{enumerate}
\item $P'=P_{ab}$, the subprofile of linear orders corresponding to the subset of vertices $V_{ab}\subseteq V$, where ${a\succ_1 b}$;
\item $A'=\{\row a{j}\}$ for some positive integer $j$ such that $k\le j\le n$;
\item the goal is to find an optimal assignment $w'\colon V_{ab}\to A'$ with $|w'(V_{ab})|\le k'$.
\end{enumerate}

The idea behind this definition is based on Lemma~\ref{termsubtr}.   Namely, we can try to guess the least preferred alternative of voter 1 in the elected committee $a_j$ and the terminal subtree $V_{ba}$ whose voters are all assigned to $a_j$, then the problem will be reduced to choosing a committee of size $k-1$ among $\{\row a{j-1}\}$ given the profile $P_{ab}$ of voters corresponding to $V_{ab}$. We note that there are at most $n-1$ subproblems. Note that our cuts can be naturally ordered: we can say that $ \text{$ab$-cut}\subset \text{$cd$-cut}$ if and only if $V_{ab}\subset V_{cd}$.\par\smallskip

%Following \cite{SYFE2013} we can prove the main theorem of this section. \par\smallskip

\noindent{\it Proof of Theorem~\ref{CCtheorem}.}
The tree $T$ with respect to which $P$ is single-crossing can be computed in polynomial time, by Theorem~\ref{conmintree}. We can identify one of the leaves then. Let this be voter~1 with preferences  $a_1\succ_1\cdots\succ_1a_m$.

For every $ab$-cut, $j\in \{1,\ldots,m\}$ and $t\in \{1,\ldots,k\}$ we define $A[V_{ab},j,t]$ to be the optimal  dissatisfaction (calculated in the utilitarian or egalitarian way) that can be achieved with a $t$-assignment function  when considering a subprofile $P'=\{P_{ab}\,|\,1\in V_{ab}\}$ over $A'=\{\row aj\}$. It is clear that for $V\subseteq V$, $j\in M\setminus \{1\}$ and $t\in K\setminus \{1\}$, the following recursive relation holds
{\small
\begin{align*}
A[V,j,t]=&\min\left\{A[V,j-1,t],\min\limits_\text{$ab$-cut}\ell(A[V_{ab},j-1,t-1],\right. \\
%&\left.(r(v_i,a_j))_{v_i\in V_{ba}})\right\}.
&(r(v_i,a_j))_{v_i\in V_{ba}})\bigg\}.
\end{align*}
}
Here $\ell$ is used to mean either the sum or the maximum of a given list of values  depending on the utilitarian or egalitarian model, respectively. 
To account for the possibility that $a_j$ is not elected in the optimal solution, we also include the term $A[V, j - 1, t]$. Thus,
$A[V,m,k]$ is then the optimal dissatisfaction, and in calculating it we simultaneously find the optimal $k$-assignment function.

The following base cases are sufficient for the recursion to be well-defined
\begin{itemize}
\item $A[\varnothing,j,t]=0$;
\item $A[V,j,1]=\min\limits_{j'\leq j} \ell((r(v_i,a_{j'}))_{i\in V})$;
\item $A[V,j,t]=0$ for $t\geq j$.  %(everybody is assigned to their most preferred alternative).
\end{itemize}

These conditions suffice for our recursion to be well-defined. 
Using dynamic programming, we can compute in polynomial time, in fact, in time $O(mn^2k)$, the optimal
dissatisfaction of the voters and the assignment that achieves it. Note that the complexity is the same as for the classical case
 in \cite{SYFE2013} which we closely followed.

\section{Conclusion}

This paper generalises the classical single-crossing property to an intermediate property on median graphs and, in particular, to trees. We complement Demange's representative voter theorem with the fact that the set of linear orders of any intermediate profile on a median graph is a Condorcet domain.

We prove that, for any median graph, there exists a profile of preferences which is intermediate with respect to that particular graph. Finally we present two results on algorithmic aspects of single-crossedness. The first one states that recognising intermediateness on median graphs is possible in polynomial time, and the second shows that the winner determination problem for the Chamberlin-Courant rule is also polynomial for single-crossing profiles on trees. We do not know whether or not this latest result can be extended to intermediate profiles on median graphs.

\bibliographystyle{named}
\bibliography{cps}

\end{document}